\documentclass[11pt]{article}

\usepackage[numbers]{natbib}

\usepackage[noend]{algorithmic,algorithm}

\usepackage{amsmath}
\usepackage{amsthm}
\usepackage{amssymb}
\usepackage{float}
\usepackage[margin=1in]{geometry}

\newtheorem{theorem}{Theorem}[section]

\newtheorem{lemma}[theorem]{Lemma}
\newtheorem{claim}[theorem]{Claim}

\newtheorem{observation}[theorem]{Observation}

\newenvironment{discussion}{\par{\noindent \bf Discussion:}}{\qed \par}
\newenvironment{proofof}[1]{\begin{proof}[Proof of #1]}{\end{proof}}

\floatstyle{ruled}
\newfloat{algorithm}{tbp}{loa}
\floatname{algorithm}{Algorithm}

\newif\iftechnicalreport
\technicalreportfalse
\technicalreporttrue

\newcommand{\comment}[1]{{ }}

\newcommand{\be}{\begin{equation}}
\newcommand{\ee}{\end{equation}}

\newcommand{\eps}{\epsilon}
\newcommand{\argmax}{\mathop{\rm argmax}}

\newcommand{\ignore}[1]{}
\newcommand{\OPT}{\textup{OPT}}
\newcommand{\AVG}{\textup{AVG}}

\def \E   {{\mathbb E}}

\def \calh {\mathcal{H}}
\def \calg {\mathcal{G}}

\begin{document}

\title{Maximizing Social Influence in Nearly Optimal Time}

\author{
Christian Borgs\thanks{Microsoft Research New England.}
\and
Michael Brautbar\thanks{Computer and Information Science, University of Pennsylvania. Now at Toast, Inc.}
\and
Jennifer Chayes\thanks{Microsoft Research New England.}
\and
Brendan Lucier\thanks{Microsoft Research New England.}
}

\date{}

\maketitle

\begin{abstract}
Diffusion is a fundamental graph process, underpinning such phenomena as epidemic disease contagion and the spread of innovation by word-of-mouth.
We address the algorithmic problem of finding a set of $k$ initial seed nodes in a network so that the expected size of the resulting cascade is maximized, under the standard \emph{independent cascade} model of network diffusion.
Runtime is a primary consideration for this problem due to the massive size of the relevant input networks.

We provide a fast algorithm for the influence maximization problem, obtaining the near-optimal approximation factor of $(1 - \frac{1}{e} - \epsilon)$, for any $\epsilon > 0$, in time $O((m+n)k\epsilon^{-2}\log n)$.
Our algorithm is runtime-optimal (up to a logarithmic factor) with respect to network size, and substantially improves upon the previously best-known algorithms which run in time $\Omega(mnk \cdot \text{POLY}(\epsilon^{-1}))$.  Furthermore, our algorithm can be modified to allow early termination: if it is terminated after $O(\beta(m+n)k\log n)$ steps for some $\beta < 1$ (which can depend on $n$), then it returns a solution with approximation factor $O(\beta)$.  Finally, we show that this runtime is optimal (up to logarithmic factors) for any $\beta$ and fixed seed size $k$.



\end{abstract}

\setcounter{page}{0}

\newpage

\section{Introduction}

Diffusion is a fundamental process in the study of complex networks, modeling the spread of disease, ideas, or product adoption through a population.  
The common feature in each case is that local interactions between individuals can lead to epidemic outcomes.  This is the idea behind word-of-mouth advertising, in which information about a product travels via links between individuals \cite{Rogers03,BrownReingen87,GoldenbergLM01,CentolaMacy07,BakshyKA09,ChaMG09,Gomez-RodriguezLK12}.
A prominent application is a viral marketing campaign which aims to use a small number of targeted interventions to initiate cascades of influence that create a global increase in product adoption \cite{DomingosRichardson01,KempeKT03,LeskovecAH07,GoldenbergLM01}.

This application gives rise to an algorithmic problem: given a network, how can we determine which individuals should be targeted
to maximize the magnitude of a resulting cascade \cite{DomingosRichardson01,RichardsonDomingos02,KempeKT03}? Supposing that there is a limit $k$ on the number of nodes to target (e.g.\ due to advertising budgets), the goal is to efficiently find an appropriate set of $k$ nodes with which to ``seed'' a diffusion process.  This problem has been studied for various models of influence spread, leading to the development of polynomial-time algorithms that achieve constant approximations \cite{KempeKT03, KempeKT05,MosselRoch07}.

Relevant networks for this problem can have massive size, on the order of billions of edges.  
The running time of an influence maximization algorithm is therefore a primary consideration.  
This is compounded by the fact that social network data and influence parameters tend to be volatile, necessitating recomputation of solutions over time.  For these reasons, near-linear runtime is a practical necessity for algorithms that work with massive network data.
This stringent runtime requirement has spawned a large body of work aimed at developing fast, heuristic methods of 
finding influential individuals in social networks, despite the existence of the above-mentioned approximation algorithms \cite{ChenWZ09,ChenWW10,ChenCWZ10,LeskovecKGFVG07,WangCSX10,MathioudakisBCGU11,JiangSCWSX11,KimuraSaito06}.  However, to date, this line of work has focused primarily on empirical methods.  Currently, the fastest algorithms with constant-factor approximation guarantees have runtime $\Omega(nmk)$ \cite{ChenWZ09}.

In this paper we bridge this gap by developing a constant-factor approximation algorithm for the influence maximization problem, under the standard \emph{independent cascade} model of influence spread, that runs in quasilinear time.   
Our algorithm can also be modified to run in sublinear time, with a correspondingly reduced approximation factor.
Before describing these results in detail, we first provide some background into the influence model.

\paragraph{The Model: Independent Cascades}
We adopt the \emph{independent cascade} (IC) model of diffusion, formalized by Kempe et al. \cite{KempeKT03}.  In this model we are given a directed edge-weighted graph $\calg$ with $n$ nodes and $m$ edges, representing the underlying network.  Influence spreads via a random process that begins at a set $S$ of seed nodes.  
Each node, once infected, has a chance of subsequently infecting its neighbors: the weight of edge $e = (v,u)$ represents the probability that the process spreads along edge $e$ from $v$ to $u$.
If we write $I(S)$ for the (random) number of nodes that are eventually infected by this process, then we think of the expectation of $I(S)$ as the \emph{influence} of set $S$.  Our optimization problem, then, is to find set $S$ maximizing $\E[I(S)]$ subject to $|S| \leq k$.

The IC model captures the intuition that influence can spread stochastically through a network, much like a disease \cite{KempeKT03,GoldenbergLM01,DoddsWatts04}.  Since its introduction, it has become one of the prominent models of influence spread; see for example \cite{ChenWZ09,ChenWW10,ChenCWZ10,LeskovecKGFVG07,WangCSX10,KempeKT05}. 
Kempe et al.\ show that $\E[I(\cdot)]$ is a submodular monotone function \cite{KempeKT03}, and hence the problem of maximizing $\E[I(\cdot)]$ can be approximated to within a factor of $(1-\frac{1}{e} -\epsilon)$ for any $\epsilon > 0$, in polynomial time, via a greedy hill-climbing method.
In contrast, many other formulations of the influence maximization problem have been shown to have strong lower bounds on polynomial-time approximability
\cite{Peleg02,Morris00,Chen08,Ben-ZewiHLN11}.

The greedy approach to maximizing influence in the IC model described above takes time $O(kn)$, \emph{given oracle access to the function $\E[I(\cdot)]$}.  However, influence values must be computed from the underlying network topology, by (for example) repeated simulation of the diffusion process.
%
This leads 
ultimately\footnote{After simple optimizations, such as reusing simulations for multiple nodes.} to a runtime\footnote{The best implementations appear to have running time $O(mnk \log(n) \cdot \text{POLY}(\epsilon^{-1}))$ \cite{ChenWZ09}, though to the best of our knowledge a formal analysis of this runtime has not appeared in the literature.}
of $\Omega(mnk \cdot \text{POLY}(\epsilon^{-1}))$.   

\paragraph{Our Result: A Quasi-Linear Time Algorithm}
Our main result is an algorithm for finding $(1-\tfrac{1}{e}-\epsilon)$-approximately optimal seed sets in arbitrary directed networks, which runs in time $O((m+n)k \epsilon^{-2}\log n)$.  
This runtime is close to optimal with respect to network size, as we give a lower bound of $\Omega(m+n)$ on the time required to obtain a constant approximation, assuming an adjacency list representation of the network as well as the the ability to uniformly sample nodes.  We also note that this approximation factor is nearly optimal, as no polytime algorithm achieves approximation $(1-\frac{1}{e}+\epsilon)$ for any $\epsilon > 0$ unless $\text{P} = \text{NP}$ \cite{KempeKT03, KempeKT05}.
Our algorithm is randomized, and it succeeds with probability $3/5$; moreover, failure is detectable, so this success probability can be amplified through repetition.  

We assume that the network topology is described in the \emph{sparse} representation of an (arbitrarily ordered) adjacency list for each vertex, as is natural for sparse graphs such as social networks. 
Our algorithms access the network structure in a very limited way: the only queries 
used by our algorithms are uniform vertex sampling and traversing the edges incident to a previously-accessed vertex. 

To describe our approach, let us first consider the problem of finding the single node with highest influence.  One strategy would be to estimate the influence of every node directly, e.g., via repeated simulation, but this is computationally expensive.
Alternatively, consider the following ``polling'' process: select a node $v$ uniformly at random, and determine the set of nodes that \emph{would have influenced $v$}.  
Intuitively, if we repeat this process multiple times, and a certain node $u$ appears often as an ``influencer,'' then $u$ is likely a good candidate for the most influential node.  Indeed, we show that the probability a node $u$ appears in a set of influencers is proportional to $\E[I(u)]$, and standard concentration bounds show that this probability can be estimated accurately with relatively few repetitions of the polling process.  Moreover, it is possible to efficiently find the set of nodes that would have influenced a node $v$: this can be done by simulating the influence process, starting from $v$, in the \emph{transpose graph} (i.e., the original network with edge directions reversed).

This motivates our algorithm, which proceeds in two steps.  First, we repeatedly apply the random sampling technique described above to generate a sparse hypergraph representation of the network.
Each hypergraph edge corresponds to a set of individuals that was influenced by a randomly selected node in the transpose graph.
This preprocessing is done once, resulting in a structure of size $O((m+n)k\epsilon^{-2}\log(n))$.  This hypergraph encodes our influence estimates: for a set of nodes $S$, the total degree of $S$ in the hypergraph is approximately proportional to the influence of $S$ in the original graph.
In the second step, we run a standard greedy algorithm on this hypergraph to return a set of size $k$ of approximately maximal total degree.

To make this approach work one needs to overcome several inherent difficulties. 
First, note that our sampling method allows us to estimate the influence of a node, but not the marginal benefit of adding a node to a partially constructed seed set.  Thus, unlike prior algorithms, we do not repeat our estimation procedure to incrementally construct a solution.  Instead, we perform all of our sampling up front and then select the entire seed set using the resulting hypergraph.

Second, our algorithm has a stringent runtime constraint --- we must construct our hypergraph in time $O((m+n)k \epsilon^{-2}\log(n))$.
To meet this bound, we must be flexible in the number of hyperedges we construct.  Instead of building a certain fixed number of hyperedges, we repeatedly build edges until the \emph{total sum of all edge sizes} exceeds $O((m+n)k \epsilon^{-2}\log(n))$.  
Intuitively speaking, this works because the number of hyperedges needed to accurately estimate influence values, times the expected size of each hyperedge, is roughly constant.  Indeed, we should expect to see large hyperedges only if the network contains many influential nodes, but high influence values require fewer samples to estimate accurately.  
%

Finally, in order to prevent errors from accumulating when we apply the greedy algorithm to the hypergraph, it is important that our estimator for the influence function (i.e.\ total hypergraph degree) is itself a monotone submodular function.


\paragraph{Early Termination and Sublinear Time}
We next show how to modify our approximation algorithm to allow early termination, providing a tradeoff between runtime and approximation quality.  Specifically, if the algorithm is allowed to run for $O(\beta (n+m)k \log(n))$ steps, and is then terminated without warning, it can immediately return a solution with approximation factor $O(\beta)$.
%
%
We also provide a lower bound of $\Omega(\max\{\beta^2, \tfrac{\beta}{k}\} \cdot (m+n) )$ on the runtime needed to obtain an $O(\beta)$-approximation.  Our algorithm is therefore nearly runtime-optimal (up to logarithmic factors) for any fixed seed size $k < \beta^{-1}$.
Our method is randomized, and it succeeds with probability $3/5$.  
As before, these results assume that the input network is provided in adjacency list format and an algorithm is allowed to perform uniform sampling of the nodes. Notably, our algorithm accesses the network structure in a very limited way and is directly implementable (with no additional cost) in a wide array of graph access models, including the ones used for sublinear-time and property testing algorithms on sparse graphs \cite{RubSha12,Goldreich10c} and the jump and crawl paradigm of \cite{BrautbarKearns10}.

The intuition behind our modified algorithm is that a tradeoff between execution time and approximation factor can be achieved by constructing fewer edges in our hypergraph representation.  Given an upper bound on runtime, we can build edges until that time has expired, then run the influence maximization algorithm using the resulting (impoverished) hypergraph.  We show that this approach generates a solution whose quality degrades gracefully with the preprocessing time, with an important caveat.  
If the network contains many nodes with high influence, it may be that a reduction in 
runtime prevents us from achieving enough concentration to estimate the influence of any node.
However, in this case, the fact that many individuals have high influence enables an alternative approach: a node chosen at random, according to the degree distribution of nodes in the hypergraph representation, is likely to have high influence.

Given the above, our algorithm will proceed by constructing two possible seed sets: one using the greedy algorithm applied to the constructed hypergraph, and the other by randomly selecting a singleton according to the hypergraph degree distribution.  If $k > 1$ we will return a union of these two solutions.  When $k=1$ we cannot use both solutions, so we must choose; 
in this case, it turns out that we can determine which of the two solutions will achieve the desired approximation by examining the maximum degree in the hypergraph.

Finally, to allow early termination without warning, the algorithm can pause its hypergraph construction and compute a tentative solution at predetermined intervals (e.g., repeatedly doubling the number of steps between computations).  Then, upon a request to terminate, the algorithm returns the most recent solution.

\subsection{Related Work}

Models of influence spread in networks, covering both cascade and threshold phenomena, are well-studied in the sociology and marketing literature \cite{Granovetter78,Rogers03,GoldenbergLM01}.  The problem of finding the most influential set of nodes to target for a diffusive process was first posed by Domingos and Richardson \cite{DomingosRichardson01,RichardsonDomingos02}.  A formal development of the IC model, along with a greedy algorithm based upon submodular maximization, was given by Kempe et al.  \cite{KempeKT03}.  Many subsequent works have studied the nature of diffusion in online social networks, using empirical data to estimate influence probabilities and infer network topology; see \cite{Liben-NowellKleniberg08,Gomez-RodriguezLK12,LeskovecMFGH07}.

It has been shown that many alternative formulations of the influence maximization problem are computationally difficult.  The problem of finding, in a \textit{linear threshold} model, a set of minimal size that influences the entire network was shown to be inapproximable within $O(n^{1-\epsilon})$ by Chen \cite{Chen08}. The problem of determining influence spread given a seed set in the IC model is \#P-hard \cite{ChenWW10}.  

There has been a line of work aimed at improving the runtime of the algorithm by Kempe et al. \cite{KempeKT03}.  These have focused largely on heuristics, such as assuming that all nodes have relatively low influence or that the input graph is clustered \cite{ChenWW10,ChenCWZ10,KimuraSaito06, WangCSX10}, as well as empirically-motivated implementation improvements \cite{LeskovecKGFVG07,ChenWZ09}.  
One particular approach of note involves first attempting to sparsify the input graph, then estimating influence on the reduced network  \cite{ChenCWZ10,MathioudakisBCGU11}.  Unfortunately, these sparsification problems are shown to be computationally intractible in general.

Various alternative formulations of influence spread as a submodular process have been proposed and analyzed in the literature \cite{MosselRoch07,KempeKT05}, including those that include interations between multiple diffusive processes \cite{GoyalKearns12,BharathiKS07}. 
We focus specifically on the IC model, and leave open the question of whether our methods can be extended to apply to these alternative models.

The influence estimation problem shares some commonality with the problems of local graph partitioning, as well as estimating pagerank and personalized pagerank vectors \cite{AndersenBCHMT2007,BorgsBCT12, SpielmanTeng04, AndersenCL06}.  These problems admit local algorithms based on sampling short random walks.  To the best of our understanding, these methods do not seem directly applicable to influence maximization. 



\section{Model and Preliminaries}\label{sec.prelim}

\paragraph{The Independent Cascade Model}
In the independent cascade (IC) model, influence spreads via an edge-weighted directed graph $\calg$. An infection begins at a set $S$ of seed nodes, and spreads through the network in rounds.  Each infected node $v$ has a single chance, upon first becoming infected, of subsequently infecting his neighbors. Each directed edge $e = (v,u)$ has a weight $p_e \in [0,1]$ representing the probability that the process spreads along edge $e$ to node $u$ in the round following the round in which $v$ was first infected. 

As noted in \cite{KempeKT03}, the above process has the following equivalent description.  We can interpret $\calg$ as a distribution over unweighted directed graphs, where each edge $e$ is independently realized with probability $p_e$.  If we realize a graph $g$ according to this probability distribution, then we can associate the set of infected nodes in the original process with the set of nodes reachable from seed set $S$ in $g$.  We will make use of this alternative formulation of the IC model throughout the paper.

\paragraph{Notation}
We let $m$ and $n$ denote the number of edges and nodes, respectively, in the weighted directed graph $\calg$.  We write $g \sim \calg$ to mean that $g$ is drawn from the random graph distribution $\calg$.  Given set $S$ of vertices and (unweighted) directed graph $g$, write $C_g(S)$ for the set of nodes reachable from $S$ in $g$.  When $g$ is drawn from $\calg$, we will refer to this as the set of nodes influenced by $S$.  We write $I_g(S) = |C_g(S)|$ for the number of nodes influenced by $S$, which we call the influence of $S$ in $g$.  We write $\E_{\calg}[I(S)] = \E_{g \sim \calg}[I_g(S)]$ for the expected influence of $S$ in $\calg$.

Given two sets of nodes $S$ and $W$, we write $C_g(S|W)$ for the set of nodes reachable from $S$ but not from $W$.  That is, $C_g(S|W) = C_g(S) \setminus C_g(W)$.  As before, we write $I_g(S|W) = |C_g(S|W)|$; we refer to this as the marginal influence of $S$ given $W$.  The expected marginal influence of $S$ given $W$ is $\E_{\calg}[I(S|W)] = \E_{g \sim \calg}[I_g(S|W)]$.

In general, a vertex in the subscript of an expectation or probability denotes the vertex being selected uniformly at random from the set of vertices of $\calg$.  For example, $\E_{v,\calg}[I(v)]$ is the average, over all graph nodes $v$, of the expected influence of $v$.

For a given graph $g$, define $g^{T}$ to be the \emph{transpose graph} of $g$: $(u,v) \in g$ iff $(v,u) \in g^{T}$.  We apply this notation to both unweighted and weighted graphs.

\paragraph{The Influence Maximization Problem}
Given graph $\calg$ and integer $k \geq 1$, the influence maximization problem is to find a set $S$ of at most $k$ nodes maximizing the value of $\E_\calg[I(S)]$.  Write $\OPT = \max_{S: |S| = k}\{\E_{\calg}[I(S)]\}$ for the maximum expected influence of any set of $k$ nodes.  For $\beta \leq 1$, we say that a particular set of nodes $S$ with $|S| \leq k$ is a $\beta$-approximation to the influence maximization problem if 
$\E_\calg[I(S)] \geq \beta \cdot \OPT$.
We assume that graph $\calg$ is provided in adjacency list format, with the neighbors of a given vertex $v$ ordered arbitrarily.

\paragraph{A Simulation Primitive}
Our algorithms we will make use of a primitive that realizes an instance of the nodes influenced by a given vertex $u$ in weighted graph $\calg$, and returns this set of nodes.  Conceptually, this is done by realizing some $g \sim \calg$ and traversing $C_g(u)$.  

Let us briefly discuss the implementation of such a primitive.  Given node $u$, we can run a depth first search in $\calg$ starting at node $u$.  Before traversing any given edge $e$, we perform a random test: with probability $p_e$ we traverse the edge as normal, and with probability $1 - p_e$ we do not traverse edge $e$ and ignore it from that point onward.  The set of nodes traversed in this manner is equivalent to $C_g(u)$ for $g \sim \calg$, due to deferred randomness.  We then return the set of nodes traversed.  The runtime of this procedure is precisely the sum of the degrees (in $\calg$) of the vertices in $C_g(u)$. 

We can implement this procedure for a traversal of $g^T$, rather than $g$, by following in-links rather than out-links in our tree traversal.

\section{An Approximation Algorithm for Influence Maximization}
\label{section:linear}

In this section we present an algorithm for the influence maximization problem on arbitrary directed graphs.  Our algorithm returns a $(1-\frac{1}{e}-\epsilon)$-approximation to the influence maximization problem, with success probability $3/5$, in time $O((m+n)k\epsilon^{-2}\log n)$.  
We discuss how to amplify this success probability in Section \ref{sec:amplify.success}.

The algorithm is described formally as Algorithm \ref{approx-exp-infl}, but let us begin by describing our construction informally.  Our approach proceeds in two steps. The first step, BuildHypergraph, stochastically generates a sparse, undirected hypergraph representation $\calh$ of our underlying graph $g$.  This is done by repeatedly simulating the influence spread process on the transpose of the input graph, $g^T$.  This simulation process is performed as described in Section \ref{sec.prelim}: we begin at a random node $u$ and proceed via depth-first search, where each encountered edge $e$ is traversed independently with probability $p_e$.  The set of nodes encountered becomes an edge in $\calh$.  We then repeat this process, generating multiple hyperedges.  The BuildHypergraph subroutine takes as input a bound $R$ on its runtime; we continue building edges until a total of $R$ steps has been taken by the simulation process.  As discussed in Section \ref{sec.prelim}, the number of steps taken by the process is equal to the number of edges considered by the depth-first search processes. 
Once $R$ steps have been taken in total over all simulations, we return the resulting hypergraph.  

In the second step, BuildSeedSet, we use our hypergraph representation to construct our output set.  This is done by repeatedly choosing the node with highest degree in $\calh$, then removing that node and all incident edges from $\calh$.  
The resulting set of $k$ nodes is the generated seed set.

\begin{algorithm}
\caption{Maximize Influence} \label{approx-exp-infl}
\begin{algorithmic}[1]
\REQUIRE Directed edge-weighted graph $\calg$, runtime bound $R$.
	\STATE $\calh \leftarrow$ BuildHypergraph$(R)$
	\RETURN BuildSeedSet$(\calh, k)$
\end{algorithmic}
\vspace{2mm}
\textbf{BuildHypergraph$(R)$:}
\begin{algorithmic}[1]
	\STATE Initialize $\calh = (V, \emptyset)$.
	\REPEAT
		\STATE Choose node $u$ from $\calg$ uniformly at random.
		\STATE Simulate influence spread, starting from $u$, in $\calg^T$.  Let $Z$ be the set of nodes discovered.		
		\STATE Add $Z$ to the edge set of $\calh$.
	\UNTIL{$R$ steps have been taken in total by the simulation process.}
	\RETURN $\calh$
\end{algorithmic}
\vspace{2mm}
\textbf{BuildSeedSet$(\calh, k)$:}
\begin{algorithmic}[1]
	\FOR {$i = 1, \dotsc, k$}
		\STATE $v_i \leftarrow \argmax_v \{ deg_{\calh}(v) \}$
		\STATE Remove $v_i$ and all incident edges from $\calh$
	\ENDFOR
	\RETURN $\{ v_1, \dotsc, v_k \}$ 
\end{algorithmic}

\end{algorithm}

\comment{
\begin{algorithm}
\caption{BuildHypergraph}
\begin{algorithmic}[1]
\REQUIRE Runtime parameter $r$.
	\STATE Initialize $\calh = (V, \emptyset)$.
	\REPEAT
		\STATE Choose node $u$ from $\calg$ uniformly at random.
		\STATE Simulate influence spread in transpose graph ${\cal}^T$ starting at $u$ by realizing a network
		$g \sim \calg$ and growing a spanning tree from $u$ in $g^{T}$.  Let $T$ be set of nodes influenced.
		\STATE Add $T$ to the edge set of $\calh$
	\UNTIL{a total of $120nk\epsilon^{-2}\log(n)$ steps have been taken by the simulation process.}
	\RETURN $\calh$
\end{algorithmic}
\end{algorithm}

\begin{algorithm}
\caption{BuildSeedSet}
\begin{algorithmic}[1]
\REQUIRE Hypegraph $\calh$, set size $k$.
	\FOR {$i = 1, \dotsc, k$}
		\STATE $v_i \leftarrow \argmax_v \{ deg_{\calh}(v) \}$
		\STATE Remove $v_i$ and all incident edges from $\calh$
	\ENDFOR
	\RETURN $\{ v_1, \dotsc, v_k \}$ 
\end{algorithmic}
\end{algorithm}
}

We now turn to provide a detailed analysis of Algorithm \ref{approx-exp-infl}.  Fix $k$ and a weighted directed graph $\calg$.  For notational convenience we will assume $m \geq n$, by adding edges of weight $0$ if necessary.  


\begin{theorem}
\label{thm.alg1.correctness}
There exists a constant $c$ such that, for any $\epsilon \in (0, \tfrac{1}{2})$, if we set $R = cmk\epsilon^{-2}\log(n)$ then Algorithm \ref{approx-exp-infl} returns a set $S$ with $\E_{\calg}[I(S)] \geq (1-\frac{1}{e}-\epsilon)\OPT$, with probability at least $3/5$.  The runtime of Algorithm \ref{approx-exp-infl} is $\Theta(R) = \Theta(mk\epsilon^{-2}\log(n))$.
\end{theorem}


The idea behind the proof of Theorem \ref{thm.alg1.correctness} is as follows.  First, we observe that the influence of a set of nodes $S$ is precisely $n$ times the probability that a node $u$, chosen uniformly at random, influences a node from $S$ in the transpose graph $g^T$.

\begin{observation}
\label{mainobs2}
For each subset of nodes $S \subseteq \calg$,
$\E_{g \sim \calg}[I_g(S)] = n \cdot {\Pr}_{u, g \sim \calg}[S \cap C_{g^T}(u) \neq \emptyset].$
\end{observation}
\begin{proof}
\[ \begin{aligned}
\E_{g \sim \calg}[I_g(S)]
&= \sum_{u \in g}{\Pr}_{g \sim \calg}[\exists v \in S \text{ such that } u \in C_g(v)] \\
&= \sum_{u \in g}{\Pr}_{g \sim \calg}[\exists v \in S \text{ such that } v \in C_{g^T}(u)] \\
&= n{\Pr}_{u,{g \sim \calg}}[\exists v \in S \text{ such that } v \in C_{g^T}(u)] \\
& = n{\Pr}_{u,{g \sim \calg}}[S \cap C_{g^T}(u) \neq \emptyset].
\end{aligned} \]
\end{proof}

Observation \ref{mainobs2} implies that we can estimate $\E_{\calg}[I(S)]$ by estimating the probability of the event $S \cap C_{g^T}(u) \neq \emptyset$.  
The degree of a node $v$ in $\calh$ is precisely the number of times we observed that $v$ was influenced by a randomly selected node $u$.  We can therefore think of $\frac{deg_{\calh}(S)}{m(\calh)}$, the fraction of edges that intersect $S$, as an estimator for $\frac{1}{n} \cdot \E_{\calg}[I(S)]$.  Our primary task is to show that it is a \emph{good} estimator.  We do this in two steps: in Lemma \ref{lem.hypergraph.size} we show that $m(\calh)$ (the number of iterations taken by BuildHypergraph) is likely to be large, then in Lemma \ref{lem.hypergraph.concentration} we bound our estimation error conditional on $m(\calh)$ being large enough.


For establishing bounds on $m(\calh)$, it will be convenient to denote by $\AVG$ the expected number of steps taken by the simulation process on any given iteration of BuildHypergraph.  That is, $\AVG$ is the expected sum of degrees (in $\calg^T$) of nodes in $C_{g^T}(u)$, where the expectation is over $g \sim \calg$ and the uniform choice of $u$.
The following technical lemma allows us to compare $\AVG$ (measured in edges) to the optimal influence (measured in nodes).

\begin{claim}
\label{claim.avg.bound}
$\AVG \leq \frac{m}{n}\OPT$.
\end{claim}
\begin{proof}
Given a vertex $u$ and an edge $e = (v,w)$, consider the random event indicating whether edge $e$ is checked as part of the process of growing a depth-first search rooted at $u$ in the IC process corresponding to graph $g^{T} \sim \calg^T$.  Note that edge $e$ is checked if and only if node $v$ is influenced by node $u$ in this invocation of the IC process.  In other words, edge $e = (v,w)$ is checked as part of the influence spread process on line 4 of BuildHypergraph if and only if $v \in Z$.
Write $m_{g^T}(u)$ for the random variable indicating the number of edges that are checked as part of building the influence set $Z$ starting at node $u$ in $g^{T}$.  We then have
\begin{align*}
\AVG 
& = \frac{1}{n} \sum_{u \in \calg} \E_{g \sim \calg}[m_{g^T}(u)]
 = \frac{1}{n} \sum_{e = (v,w) \in \calg^T} \E_{g \sim \calg }\left[|\{ u : v \in C_{g^T}(u) \}| \right].
\end{align*}
Noting that $v \in C_{g^T}(u)$ if and only if $u \in C_{g^T}(v)$, we then have
\begin{align*}
\AVG
& = \frac{1}{n} \sum_{e = (v,w) \in \calg^T} \E_{g \sim \calg }\left[|\{ u : u \in C_{g}(v) \}| \right] \\
& = \frac{1}{n} \sum_{e = (v,w) \in \calg^T} \E_{\calg}[I(\{v\})] \\
& \leq \frac{1}{n} \sum_{e = (v,w) \in \calg^T} \OPT
 = \frac{m}{n} \OPT
\end{align*}
as required.
\end{proof}



%

We will think of $\calh$ as being built incrementally, one edge at a time, as in the execution of BuildHypergraph.
Since we will discuss the state of $\calh$ at various points of execution, we will write $M$ for the number of edges in $\calh$ when BuildHypergraph terminates.  Note that $M$ as a random variable, corresponding to $m(\calh)$ after BuildHypergraph has completed.

\begin{lemma}
\label{lem.hypergraph.size}
With probability at least $1 - \frac{2}{n^{ck/4(1+\epsilon)}}$, $M$ satisfies $\frac{1}{1+\epsilon} \cdot \frac{R}{\AVG} \leq M \leq \frac{1}{1-\epsilon} \cdot \frac{R}{\AVG}$.
\end{lemma}

\begin{proof}
Write $x_i$ for the fraction of all edges in $\calg$ that are checked as part of the $i$th iteration of BuildHypergraph in Algorithm \ref{approx-exp-infl}.  That is, when building the $i$th hyperedge of $\calh$.  Note then that $x_i \in [0,1]$ for all $i$, and $\E[x_i] = \frac{\AVG}{m}$ for every $i$.  Moreover, $M$ is precisely the minimal value $J$ such that $\sum_{i=1}^J x_i \geq R/m$.  

Let $J_1 = \frac{1}{1+\epsilon} \cdot \frac{R}{\AVG}$ and $J_2 = \frac{1}{1-\epsilon} \cdot \frac{R}{\AVG}$.
We wish to show that, with high probability, $M \in [J_1, J_2]$.  This is equivalent to the intersection of events $\sum_{i=1}^{J_1} x_i < R/m$ and $\sum_{i=1}^{J_2} x_i > R/m$.  For the first event, note that
\begin{equation}
\label{eq.j1.sum}
\E\left[ \sum_{i=1}^{J_1} x_i \right] = \sum_{i=1}^{J_1} \E[x_i] = J_1 \cdot \frac{\AVG}{m} = \frac{R}{m(1+\epsilon)}.
\end{equation}
Chernoff bounds then imply
\begin{equation}
\label{eq.j1.concentrate}
\Pr\left[ \sum_{i=1}^{J_1} x_i \geq R/m \right] 
= \Pr\left[ \sum_{i=1}^{J_1} x_i \geq (1+\epsilon)\E\left[\sum_{i=1}^{J_1} x_i\right] \right]
< e^{-\epsilon^2 \cdot \E\left[ \sum_{i=1}^{J_1} x_i \right] / 4}.
\end{equation}
Substituting \eqref{eq.j1.sum} into \eqref{eq.j1.concentrate} and using $R = cmk\epsilon^{-2}\log(n)$, we obtain
\[ 
\Pr\left[ \sum_{i=1}^{J_1} x_i \geq R/m \right] 
< e^{-\epsilon^2 \cdot (1+\epsilon)^{-1}ck\epsilon^{-2}\log(n)/4}
< \frac{1}{n^{ck/4(1+\epsilon)}}.
\]
Applying an identical argument for $J_2$, we also have 
\begin{align*}
\Pr\left[ \sum_{i=1}^{J_2} x_i \leq R/m \right] < \frac{1}{n^{ck/2(1-\epsilon)}} < \frac{1}{n^{ck/4(1+\epsilon)}}.
\end{align*}
Taking a union bound over the two events yields the desired bound.
%
%
\end{proof}

%
%

We are now ready to show that the resulting hypergraph estimates the influence of each set within a sufficiently small error.
For a given integer $J$, Let $D_S^J$ be a random variable denoting the degree of $S$ in $\calh$ after $J$ edges have been added to $\calh$.  We will allow $J > M$, by considering the (hypothetical) path of execution of BuildHypergraph had there not been a runtime bound.

\begin{lemma}
\label{lem.hypergraph.concentration}
For any $J \geq \frac{1}{1+\epsilon} \cdot \frac{R}{\AVG}$ and set of nodes $S \subseteq V$ with $|S| \leq k$, 
\begin{equation}
\label{eq.concentration}
\Pr\left[D_S^{J} > \E[D_S^{J}] + \frac{\epsilon \cdot J \cdot \OPT}{n}\right] < \frac{1}{n^{ck/4(1+\epsilon)}}
\end{equation}
and
\begin{equation}
\label{eq.concentration.2}
\Pr\left[D_S^{J} < \E[D_S^{J}] - \frac{\epsilon \cdot J \cdot \OPT}{n}\right] < \frac{1}{n^{ck/2(1+\epsilon)}}.
\end{equation}
\end{lemma}


\begin{proof}
%
%
%
We first show \eqref{eq.concentration}.
Observation \ref{mainobs2} implies that $D_S^J$ is the sum of $J$ identically distributed Bernoulli random variables, each with probability $\E_{\calg}[I(S)]/n$.  Note that we must have $\E[D_S^J] \leq J \cdot \frac{\OPT}{n}$, with equality occuring for the set $S$ of maximal influence.
%
%
For notational convenience, choose $\lambda \in (0,1]$ so that $\lambda \cdot \E[D_S^J] = \frac{J \cdot \OPT}{n}$.
Then 
\[ \Pr\left[D_S^{J} > \E[D_S^{J}] + \frac{\epsilon \cdot J \cdot \OPT}{n}\right] = \Pr\left[ D_S^J > \left(1 + \frac{\epsilon}{\lambda}\right) \E\left[ D_S^J \right] \right]. \]
We consider two cases.  First, if $\lambda > \epsilon$, then $\epsilon / \lambda < 1$, and the Multiplicative Chernoff bound (\ref{lemma.multchernoff})  implies that 
\begin{align*}
\Pr\left[ D_S^J > \left(1 + \frac{\epsilon}{\lambda}\right) \E\left[ D_S^J \right] \right]
& < e^{-\epsilon^2 \cdot \E[D_S^J]/4\lambda^2 }.
\end{align*}
Using the fact that $\E[D_S^J] = \frac{J \cdot \OPT}{\lambda n}$ and $J \geq \frac{1}{1+\epsilon} \cdot \frac{R}{\AVG} = \frac{c m k \epsilon^{-2} \log n }{\AVG(1+\epsilon)}$, we obtain
\begin{align*}
\Pr\left[ D_S^J > \left(1 + \frac{\epsilon}{\lambda}\right) \E\left[ D_S^J \right] \right]
& < e^{-\frac{m \OPT c k \log n}{4(1+\epsilon)\lambda^3 n \AVG } }
\leq e^{- c k \log n / 4(1+\epsilon) } = \frac{1}{n^{ck/4(1+\epsilon)}}
\end{align*}
where the second inequality used Claim \ref{claim.avg.bound} and $\lambda \leq 1$.  

Next suppose $\lambda \leq \epsilon$.  Then $\epsilon / \lambda \geq 1$, and the Chernoff bound implies
\begin{align*}
\Pr\left[ D_S^J > \left(1 + \frac{\epsilon}{\lambda}\right) \E\left[ D_S^J \right] \right]
& < e^{-\epsilon \cdot \E[D_S^J]/3\lambda }.
\end{align*}
Again using $\E[D_S^J] = \frac{J \cdot \OPT}{\lambda n}$ and $J \geq \frac{c m k \epsilon^{-2} \log n }{\AVG(1+\epsilon)}$, we obtain
\begin{align*}
\Pr\left[ D_S^J > \left(1 + \frac{\epsilon}{\lambda}\right) \E\left[ D_S^J \right] \right]
& < e^{-\frac{m \OPT c k \log n }{3 (1+\epsilon) \lambda^2 n \AVG} }
< e^{- c k \log n / 3 (1+\epsilon) }
= \frac{1}{n^{c k / 3 (1+\epsilon)}}
\end{align*}
where the second inequality used Claim \ref{claim.avg.bound} and $\lambda \leq 1$.  This completes the derivation of \eqref{eq.concentration}.  The derivation of \eqref{eq.concentration.2} is similar, yielding
\begin{align*}
\Pr\left[D_S^{J} > \E[D_S^{J}] - \frac{\epsilon \cdot J \cdot \OPT}{n}\right] 
&= \Pr\left[ D_S^J > \left(1 - \frac{\epsilon}{\lambda}\right) \E\left[ D_S^J \right] \right]\\
&< \frac{1}{n^{ck/2(1+\epsilon)}}
\end{align*}
as required.
\end{proof}

Finally, we must show that the greedy algorithm applied to $\calh$ in BuildSeedSet returns a good approximation to the original optimization problem.  Recall that, in general, the greedy algorithm for submodular function maximization proceeds by repeatedly selecting the singleton with maximal contribution to the function value, up to the cardinality constraint.  The following lemma shows that if one submodular function is approximated sufficiently well by a distribution of submodular functions, then applying the greedy algorithm to a function drawn from the distribution yields a good approximation with respect to the original.

\begin{lemma}
\label{lem.submod-approx}
Choose $\delta > 0$ and suppose that $f \colon 2^V \to \mathbb{R}_{\geq 0}$ is an arbitrary 
function.  Let $D$ be a distribution over 
functions with the property that, for all sets $S$ with $|S| \leq k$, $\Pr_{\hat{f} \sim D}[|f(S) - \hat{f}(S)| > \delta] < 1/n^{k+\ell}$.  If we write $S_{\hat{f}}$ for the set returned by the greedy algorithm on input $\hat{f}$, then
\[ \Pr_{\hat{f} \sim D}\left[f(S_{\hat{f}}) < (1 - 1/e)\left(\max_{S \colon |S| = k}f(S)\right) - 2\delta\right] < 1/n^\ell. \]
\end{lemma}
\begin{proof}
%
Taking a union bound over all sets of size at most $k$, of which there are at most $n^k$, we have
\[ \Pr[ \exists S \colon |\hat{f}(S) - f(S)| > \delta ] \leq n^k \cdot n^{-k -\ell} = 1 / n^{\ell}. \]
Conditioning on the event that $|\hat{f}(S) - f(S)| \leq \delta$ for every set $S$ of size at most $k$, we then have
\[ f(S_{\hat{f}} ) \geq \hat{f}(S_{\hat{f}}) - \delta \geq (1-1/e)\hat{f}(S^*) - \delta \geq (1-1/e)f(S^*) - 2\delta \]
as required.
\end{proof}

We are now ready to complete our proof of Theorem \ref{thm.alg1.correctness}.


\begin{proofof}{Theorem \ref{thm.alg1.correctness}}
As in the proof of Lemma~\ref{lem.hypergraph.size}, we will write $J_1 = \frac{1}{1+\epsilon} \cdot \frac{R}{\AVG}$ and $J_2 = \frac{1}{1-\epsilon} \cdot \frac{R}{\AVG}$.
Condition on the event that $M \in [J_1, J_2]$, which (by Lemma \ref{lem.hypergraph.size}) has probability at least $1 - \frac{2}{n^{ck/4(1+\epsilon)}}$.  We know from Lemma \ref{lem.hypergraph.concentration} that, with probability at least $1 - \frac{1}{n^{ck/4(1+\epsilon)}}$, $|D_S^{J_2} - \E[D_S^{J_2}]| < \frac{\epsilon \cdot J_2 \cdot OPT}{n}$; we will condition on this event as well.  Since $D_S^M$ is dominated by $D_S^{J_2}$, and since $J_2 < \tfrac{1+\eps}{1-\eps}J_1 \leq (1+4\eps)M$ for $\epsilon \leq 1/2$, we conclude that
\begin{align*} 
D_S^M \leq D_S^{J_2} & < \E[D_S^{J_2}] + \frac{\epsilon OPT}{n}J_2 = \E[D_S^M] \cdot \frac{J_2}{M} + \frac{\epsilon OPT}{n}J_2 \\
& < \E[D_S^M] + 4\eps \cdot \E[D_S^M] + \frac{\epsilon OPT}{n}M \\
& < \E[D_S^M] + \frac{5 \epsilon OPT}{n}M.
\end{align*}
Similarly, the fact that $D_S^M$ dominates $D_S^{J_1}$ implies $D_S^M > \E[D_S^M] - \frac{5 \epsilon OPT}{n}M$, conditioning on an event with probability at least $1 - \frac{1}{n^{ck/4(1+\epsilon)}}$.  Taking the union bound over the complement of the conditioned events, we conclude that 
$$\Pr\left[ \left|\E_{\calg}[I(S)] - \frac{n \cdot deg_{\calh}(S)}{m(\calh)} \right| > 5\epsilon \OPT \right] < \frac{4}{n^{ck/4(1+\epsilon)}}$$ for each $S \subseteq V$.  We then apply Lemma \ref{lem.submod-approx} with $f(S) := \E_{\calg}[I(S)]$, $\hat{f}(S) := \frac{n \cdot deg_{\calh}(S)}{m(\calh)}$ (drawn from then distribution corresponding to distribution of $\calh$ returned by BuildHypergraph), and $\delta = 5 \eps \OPT$.  
As long as $\tfrac{ck}{4(1+\epsilon)} \geq k+1$, Lemma \ref{lem.submod-approx} implies that, with probability at least $1 - \frac{1}{n}$, the greedy algorithm applied to $\calh$ returns a set $S$ with $\E_{\calg}[I(S)] \geq (1 - \tfrac{1}{e})\OPT - 10\eps \OPT = (1 - \tfrac{1}{e} - 10\eps)\OPT$. Noting that this is precisely the set returned by BuildSeedSet gives the desired bound on the approximation factor (rescaling $\epsilon$ by a constant factor). Thus the claim holds with probability at least $2/3 -1/n \geq 3/5$ (for $n \geq 20$).

Our condition $\tfrac{ck}{4(1+\epsilon)} \geq k+1$ is satisfied for $c = 4(1+\epsilon)(1 + \frac{1}{k})$.  We can therefore take this value of $c$ in our definition of $R$.

Finally, we argue that our algorithm can be implemented in the appropriate runtime.
The fact that BuildHypergraph executes in the required time follows from the explicit bound on its runtime.
For BuildSeedSet, we will maintain a list of vertices sorted by their degree in $\calh$; this will allow us to repeatedly select the maximum-degree node in constant time.  The initial sort takes time $O(n\log n)$.  We must bound the time needed to remove an edge from $\calh$ and correspondingly update the sorted list.  We will implement the sorted list as a doubly linked list of groups of vertices, where each group itself is implemented as a doubly linked list containing all vertices of a given degree (with only non-empty groups present).  Each edge of $\calh$ will maintain a list of pointers to its vertices.  When an edge is removed, the degree of each vertex in the edge decreases by $1$; we modify the list by shifting any decremented vertex to the preceding group (creating new groups and removing empty groups as necessary).  Removing an edge from $\calh$ and updating the sorted list therefore takes time proportional to the size of the edge.
Since each edge in $\calh$ can be removed at most once over all iterations of BuildSeedSet, the total runtime is at most the sum of node degrees in $\calh$, which is at most $R$. 
\end{proofof}

\subsection{Amplifying the Success Probability}
\label{sec:amplify.success}

Algorithm \ref{approx-exp-infl} returns a set of influence at least $(1-\frac{1}{e}-\epsilon)$ with probability at least $3/5$.  The failure probability is due to Lemma \ref{lem.hypergraph.size}: hypergraph $\calh$ may not have sufficiently many edges after $R$ steps have been taken by the simulation process in line 4 of the BuildHypergraph subprocedure.  However, note that this failure condition is detectable via repetition:  we can repeat Algorithm \ref{approx-exp-infl} multiple times, and use only the iteration that generates the most edges.  The success rate can then be improved by repeated invocation, up to a maximum of $1 - 1/n$ with $\log(n)$ repetitions (at which point the error probability due to Lemma \ref{lem.hypergraph.concentration} becomes dominant). 

We next note that, for any $\ell > 1$, the error bound in Lemma \ref{lem.hypergraph.concentration} can be improved to $\frac{1}{n^\ell}$, by increasing the value of $R$ by a factor of $\ell$, since this error derives from Chernoff bounds.  This would allow the success rate of the algorithm to be improved up to a maximum of $1 - \frac{1}{n^\ell}$ by further repeated invocation.  To summarize, the error rate of the algorithm can be improved to $1 - \frac{1}{n^\ell}$ for any $\ell$, at the cost of increasing the runtime of the algorithm by a factor of $\ell^2 \log(n)$.

\section{Approximate Influence Maximization in Sublinear Time}\label{section:sublinear}

We now describe a modified algorithm that provides a tradeoff between runtime and approximation quality.  For an an arbitrary $\beta < 1$, our algorithm will obtain an $O(\beta)$-approximation to the influence maximization problem, in time $O(\beta (n+m) k \log(n))$, with probability at least $3/5$. 
In Section \ref{sec:dynamic} we describe an implementation of this algorithm that supports termination after an arbitrary number of steps, rather than being given the value of $\beta$ in advance.  

Our algorithm is listed as Algorithm \ref{approx-exp-infl-rho}.  The intuition behind our construction is as follows.  We wish to find a set of nodes with high expected influence.  One approach would be to apply Algorithm \ref{approx-exp-infl} and simply impose a tighter constraint on the amount of time that can be used to construct hypergraph $\calh$.  This might correspond to reducing the value of parameter $R$ by, say, a factor of $\beta$.  Unfortunately, the precision of our sampling method does not always degrade gracefully with fewer samples: if $\beta$ is sufficiently small, we may not have enough data to guess at a maximum-influence node (even if we allow ourselves a factor of $\beta$ in the approximation ratio).  In these cases, the sampling approach fails to provide a good approximation.

However, as we will show, our sampling fails precisely because many of the edges in our hypergraph construction were large, and (with constant probability) this can occur only if many of the nodes that make up those edges have high influence.  
In this case, we could proceed by selecting a node from the hypergraph at random, with probability proportional to its hypergraph degree.  We prove that this procedure is likely to return a node of very high influence precisely in settings where the original sampling approach would fail.

If $k > 1$, we can combine these two approaches by returning a union of vertices selected according to each procedure.  If $k=1$, we must choose which approach to apply.  However, in this case, there is a simple way to determine whether we have obtained enough samples that BuildSeedSet returns an acceptable solution: check whether the maximum degree in the hypergraph is sufficiently high.  



%

\begin{algorithm}
\caption{Runtime-Flexible Influence Maximization} \label{approx-exp-infl-rho}


\begin{algorithmic}[1]
\REQUIRE Directed edge-weighted graph $\calg$, runtime bound $R$, threshold $C$.
\vspace{1mm}
	\STATE $\calh \leftarrow$ BuildHypergraph$(R)$
	\STATE Choose $v \in V$ with probability proportional to degree in $\calh$
	\IF{$k > 1$}
		\STATE $S \leftarrow$ BuildSeedSet$(\calh,k-1)$
		\RETURN $S \cup \{v\}$
	\ELSE
		\STATE $S \leftarrow$ BuildSeedSet$(\calh,1)$
		\STATE \textbf{if} $\max_u \{deg_{\calh}(u)\} > C$ \textbf{then} \textbf{return} $S$
		\STATE \textbf{else} \textbf{return} $\{v\}$
	\ENDIF

\end{algorithmic}


\end{algorithm}

\begin{theorem}
\label{thm.exp.infl.rho}
There exist constants $c_1, c_2$ such that, for any $\beta \in (0,1)$, if we set $R = c_1 c_2 \beta m k \log(n)$ and $C = c_2 \log(n)$, then Algorithm \ref{approx-exp-infl-rho} returns a set $S$ with $\E_{\calg}[I(S)] \geq \min\{ \frac{1}{4}, \beta \} \cdot \OPT$, with probability at least $3/5$.  The runtime of Algorithm \ref{approx-exp-infl-rho} is $\Theta(R) = \Theta(\beta m k  \log(n))$.
\end{theorem}

Our proof of Theorem \ref{thm.exp.infl.rho} proceeds via two cases, depending on whether 
$\calh$ has many or few edges as a function of $\OPT$.  
We first show that, subject to 
$\calh$ having many edges, 
set $S$ from line $5$ or $8$ (corresponding to $k>1$ and $k=1$, respectively) is likely to have high influence.  This follows the analysis from Theorem \ref{thm.alg1.correctness}. 



\begin{lemma}
\label{lem.influence1}
Suppose constant $c_2$ in the statement of Theorem \ref{thm.exp.infl.rho} is sufficiently large.  
If $m(\calh) \geq \frac{c_2 n k \log(n)}{\OPT}$, 
then the set $S$ returned by Algorithm \ref{approx-exp-infl-rho} satisfies $\E_{\calg}[I(S)] \geq \frac{1}{4}\OPT$ with probability at least $1 - \frac{1}{n}$ (over randomness in $\calh$).
\end{lemma}
\begin{proof}
%
Suppose $k=1$, so $S$ is as defined on line $7$.  
If we now apply Lemma \ref{lem.hypergraph.concentration} to $J = \frac{c_2 n k \log(n)}{\OPT}$, taking $\epsilon = \tfrac{1}{6}$,
followed by the analysis of BuildSeedSet$(\calh,k)$ from the proof of Theorem \ref{thm.alg1.correctness}, we get that $\E_{\calg}[I(S)] \geq \frac{1}{2}\OPT$, as required.

If  $k > 1$, applying Lemma \ref{lem.hypergraph.concentration} with $\epsilon = \frac{1}{6}$ yields instead that $\E_{\calg}[I(S)] \geq \frac{1}{2}\OPT_{k-1}$, where $\OPT_{k-1}$ is the maximum influence over sets of size at most $k-1$.  But now, by submodularity, $\frac{1}{2}\OPT_{k-1} \geq \frac{1}{2}(\frac{k-1}{k})\OPT \geq \frac{1}{4}\OPT$, as required.
\end{proof}

Note that $\beta$ does not appear explicitly in the statement of Lemma \ref{lem.influence1}.  The (implicit) role of $\beta$ in Lemma \ref{lem.influence1} is that as $\beta$ becomes small, Algorithm \ref{approx-exp-infl-rho} uses fewer steps to construct hypergraph $\calh$ and hence the condition of the lemma is less likely to be satisfied.
We next show that if $m(\calh)$ is small, then node $v$ from line $3$ is likely to have high influence.  This follows because, in a small number of edges, we do not expect to see many nodes with low influence.  Since we see a large number of nodes in total, we conclude that most of them must have high influence.


\begin{lemma}
\label{lem.influence2}
Let $c_1$ and $c_2$ be the constants from the statement of Theorem \ref{thm.exp.infl.rho} and suppose $c_1$ is sufficiently large.  If 
$m(\calh) < \frac{4 c_2 n k \log(n)}{OPT}$, 
then node $v$ (from line 3 of Algorithm \ref{approx-exp-infl-rho}) satisfies $\E_{\calg}[I(v)] \geq \beta\cdot\OPT$ with probability at least $2/3$ (over randomness in $\calh$).
\end{lemma}

\begin{proof}
Let random variable $X$ denote the number of times that a node with influence at most $\beta \cdot \OPT$ was added to a hyperedge of $\calh$.  Since $\calh$ has fewer than $\frac{4 c_2  n k \log(n)}{\OPT}$ edges, the expected value of $X$ is at most
\begin{eqnarray*}
E[X] &\leq& \frac{4 c_2 n k \log(n)}{\OPT} \sum_{u \in V}\frac{1}{n}\min\{\E_{\calg}[I(u)], \beta \cdot \OPT\} \\ 
&\leq& 4 c_2 \beta n k \log(n).
\end{eqnarray*}
Markov inequality then gives that $\Pr[X > 24 c_2 \beta n k \log(n)] < 1/6$.  Conditioning on this event, we have that at most $24 c_2 \beta n k \log(n)$ of the nodes touched by BuildHypergraph (counted with multiplicity) have influence less than $\beta \cdot \OPT$.  Note that drawing node $v$ with probability proportional to its degree is equivalent to drawing uniformly from all events in which a node is touched by BuildHypergraph.  Thus, since at least $c_1 c_2 \beta n k \log(n)$ nodes were touched in total, the probability that node $v$ from line $4$ has influence less than $\beta \cdot \OPT$ is at most $24/c_1$.  As long as $c_1$ is sufficiently large, this is at most $1/6$.  The union bound then allows us to conclude that $v$ has $\E[I(v)] \geq \beta \cdot \OPT$ with probability at least $1-(1/6 + 1/6) \geq 2/3$.
\end{proof}

For the case $k=1$, the algorithm chooses between returning $S$ and returning $\{v\}$, based on the maximum degree in $\calh$.  The following lemma motivates this choice.  The proof follows from an application of concentration bounds: if a node is present in $O(\log n)$ hyperedges, then with high probability we have obtained enough samples to accurately estimate its influence.

\begin{lemma}
\label{lem.k1}
If $k=1$ then the following is true with probability at least $1 - \frac{2}{n}$.  If $\max_u\{deg_\calh(u)\} > 2 c_2 \log n$ then 
$m(\calh) > \frac{c_2 n \log(n)}{OPT}$.
Otherwise, $m(\calh) < \frac{4 c_2 n \log(n)}{OPT}$.
\end{lemma}
\begin{proof}
As in the proof of Lemma \ref{lem.hypergraph.concentration}, we will think of $\calh$ as being built incrementally edge by edge, and we will let $D_w^J$ denote the degree of vertex $w$ in $\calh$ after $J$ edges have been added.  Then, for any fixed $J$, $D_w^J$ is precisely the sum of $J$ Bernoulli random variables, each with expectation $\frac{1}{n}\E_{\calg}[I(w)]$, and hence $\E[D_w^J] = \frac{J}{n}\E_{\calg}[I(w)]$.

Let $J_1$ be the smallest value such that $\E[D_w^{J_1}] > 4 c_2 \log n$.  Chernoff bounds 
imply
\[ \Pr[ D_w^{J_1} \leq 2 c_2 \log n ] < e^{-2\log n} = 1/n^2. \]
Suppose that this event does not occur.  Then for any $J \geq J_1$, we have $D_w^J > D_w^{J_1} > 2 c_2 \log n$.  In particular, if $\E[D_w^{m(\calh)}] > 4 c_2 \log n$, we must have $m(\calh) \geq J_1$, and hence $D_w^{m(\calh)} > 2 c_2 \log n$.

Let $J_2$ be the largest value such that $\E[D_w^{J_2}] < c_2 \log n$.  Chernoff bounds 
again imply that
\[ \Pr[ D_w^{J_2} \geq 2 c_2 \log n] < e^{-2\log n} = 1/n^2. \]
Suppose that this event does not occur.  Then for any $J \leq J_2$, we have $D_w^J < D_w^{J_2} < 2 c_2 \log n$.  In particular, for any $w$ such that $\E[D_w^{m(\calh)}] < c_2 \log n$, we must have $m(\calh) \leq J_2$, and hence $D_w^{m(\calh)} < 2 c_2 \log n$.

Taking the union bound over all $w$, we conclude that with probability at least $1 - 2/n$, only $w$ for which $\E[D_w^{m(\calh)}] > c_2 \log n$, and every $w$ with $\E[D_w^{m(\calh)}] > 4 c_2 \log n$, will have $D_w^{m(\calh)} \geq 2 c_2 \log n$.  We will condition on this event for the remainder of the proof.

Suppose that $\max_w D_w^{m(\calh)} < 2 c_2 \log n$.  Then we have that $\max_w \E[D_w^{m(\calh)}] < 4 c_2 \log n$.  Since $\max_w \E[D_w^{m(\calh)}] = \max_w \frac{1}{n}m(\calh)\E_{\calg}[I(w)] = \frac{1}{n}m(\calh)\cdot OPT$, we conclude $m(\calh) < \frac{4 c_2 n \log(n)}{OPT}$ as required.

Next suppose that $\max_w D_w^{m(\calh)} > 2 c_2 \log n$.  We then have that $\max_w \E[D_w^{m(\calh)}] > c_2 \log n$.  Since, again, $\max_w \E[D_w^{m(\calh)}] = \frac{1}{n}m(\calh)\cdot OPT$, we conclude $m(\calh) > \frac{c_2 n\log(n)}{OPT}$ as required.
\end{proof}

We are now ready to complete the proof of Theorem \ref{thm.exp.infl.rho}.

\begin{proofof}{Theorem \ref{thm.exp.infl.rho}}
Lemma \ref{lem.influence1} and Lemma \ref{lem.influence2} imply that,
with probability at least $2/3-1/n^2 \geq 3/5$ (for $n \ge 5$), one of $S$ or $\{v\}$ has influence at least $\text{min} \{ \frac{1}{4}, \beta \} \cdot OPT$, and therefore $S \cup \{v\}$ does as well.  If $k > 1$ then we return $S \cup \{v\}$ and we are done.  Otherwise, Lemma \ref{lem.k1} implies that if we return set $S$ then the influence of $S$ is at least $OPT/4$ (by Lemma \ref{lem.influence1}), and if we return set $v$ then the expected influence of $v$ is at least $\beta \cdot OPT$ (by Lemma \ref{lem.influence2}).  Thus, in all cases, we return a set of influence at least $\text{min} \{ \frac{1}{4}, \beta \} \cdot OPT$.
The required bound on the runtime of Algorithm~\ref{approx-exp-infl-rho} follows directly from the value of $R$, as in the proof of Theorem \ref{thm.alg1.correctness}.
\end{proofof}

\subsection{Dynamic Runtime}
\label{sec:dynamic}

Algorithm \ref{approx-exp-infl-rho} assumes that the desired approximation factor, $\beta$, is provided as a parameter to the problem.  We note that a slight modification to the algorithm removes the requirement that $\beta$ be specified in advance.  That is, we obtain an algorithm that can be terminated without warning, say after $O(\gamma \cdot (n+m) k \log(n))$ steps for some $\gamma \leq 1$, at which point it immediately returns a solution that is an $O(\gamma)$ approximation with probability at least $\frac{3}{5}$.  To achieve this, we execute Algorithm \ref{approx-exp-infl-rho} as though $\beta = 1$, but then modify BuildHypergraph so that, for each $i \geq 1$, we pause the creation of hypergraph $\calh$ after $2^i$ steps and complete the algorithm using the current hypergraph, which takes time at most $O(2^i)$.  Once this is done, we save the resulting solution and resume the creation of the hypergraph until the next power of $2$.  When the algorithm is terminated, we return the most recently-computed solution; this corresponds to a solution for a hypergraph built using at least half of the total steps taken by the algorithm at the time of termination.  Theorem \ref{thm.exp.infl.rho} then implies that this solution has approximation $O(\gamma)$ if termination occurs after $O(\gamma \cdot (n+m) k \log(n))$ steps.

\subsection{A Lower Bound}
We provide a lower bound on the time it takes for any algorithm, equipped with uniform node sampling, to compute a $\beta$-approximation for the maximum expected influence problem under the adjacency list network representation.  In particular, for any given budget $k$, at least $\Omega(\beta n)$ queries are required to obtain approximation factor $\beta$ with fixed probability.

\begin{theorem}
\label{thm.lowerbound}
Let $0 < \epsilon <\frac{1}{10e}$, $\beta \leq 1 $ be given.  Any randomized algorithm for the maximum influence problem that has runtime of $\frac{\beta(m+n)}{24\min\{k, 1/\beta\}}$ cannot return, with probability at least $1-\frac{1}{e}-\epsilon$, a set of nodes with approximation ratio better than $\beta$.
\end{theorem}

\begin{proof}
Note first that for a graph consisting of $n$ singletons, an algorithm must return at least $\beta k$ nodes to obtain an approximation ratio of $\beta$.  Doing so in at most $\beta^2 n/2$ queries requires that $2\beta k \leq \beta^2 n$, which implies $2k/\beta \leq n$.  We can therefore assume $2k/\beta \leq n$. 

The proof will use Yao's Minimax Principle for the performance of Las Vegas (LV) randomized algorithms on a family of inputs~\cite{Yao77}. The lemma states that the least expected cost of deterministic LV algorithms on a distribution over a family inputs is a lower bound on the expected cost of the optimal randomized LV algorithm over that family of inputs. 
Define the cost of the algorithm as $0$ if it returns a set nodes with approximation ratio better than $\beta$ and $1$ otherwise. As the cost of an algorithm equals its probability of failure, we can think of it as a LV algorithm.

Assume for notational simplicity that $\beta = 1/T$ where $T$ is an integer. We will build a family of lower bound graphs, one for each value of $n$ (beginning from $n = 1+T$); each graph will have $m \leq n$, so it will suffice to demonstrate a lower bound of $\frac{n}{12 T \min\{k,T\}}$.

We now consider the behavior of a deterministic algorithm $A$ with respect to the uniform distribution on the constructed family of inputs.  
For a given value $T$ the graph would be made from $k$ components of size $2T$ and $n-2kT$ singleton components (recall that $2kT = 2k/\beta \leq n$).  
If algorithm $A$ returns nodes from $\ell$ of the $k$ components of size $2T$, it achieves a total influence of $2\ell T+(k-\ell)$.  Thus, to attain approximation factor better than $\beta = \frac{1}{T}$, we must have $2\ell T+(k-\ell) \geq \frac{1}{T}2kT$, which implies $\ell \geq \frac{k}{2T - 1}$ for any $T > 1$.  

Suppose $k > 12T$.  The condition $\ell \geq \frac{k}{2T-1}$ implies that at least $\frac{k}{2T-1}$ of the large components must be queried by the algorithm, where each random query has probability $\frac{2kT}{n}$ of hitting a large component.  If the algorithm makes fewer than $\frac{n}{12T^2}$ queries, then the expected number of components hit is $\frac{n}{12T^2} \cdot \frac{2kT}{n} = \frac{k}{6T}$.  The Multiplicative Chernoff bound (Lemma \ref{lemma.multchernoff}, part 3) then imply that the probability hitting more than $\frac{k}{2T}$ components is no more than $e^{-\frac{k}{6T}\cdot 2 / 3} \leq \frac{1}{e^{4/3}} < 1-\frac{1}{e}-\epsilon$, a contradiction.

If $k \leq 12T$ then we need that $\ell \geq 1$, which occurs only if the algorithm queries at least one of the $kT$ vertices in the large components.  With $\frac{n}{2kT}$ queries, for $n$ large enough, this happens with probability smaller than $\frac{1}{e}- \epsilon$, a contradiction.

We conclude that, in all cases, at least $\frac{n}{12T\min\{k,T\}}$ queries are necessary to obtain approximation factor better than $\beta = \frac{1}{T}$ with probability at least $1-\frac{1}{e} -\epsilon$, as required.

By Yao's Minimax Principle this gives a lower bound of $\Omega(\frac{nd}{24T\min\{k, T\}})$ on the expected performance of any randomized algorithm, on at least one of the inputs.

Finally, the construction can be modified to apply to non-sparse networks.  For any $d \leq n$, we can augment our graph by overlaying a $d$-regular graph with exponentially small weight on each edge.  This does not significantly impact the influence of any set, but increases the time to decide if a node is in a large component by a factor of $O(d)$ (as edges must be traversed until one with non-exponentially-small weight is found).  Thus, for each $d \leq n$, we have a lower bound of $\frac{nd}{24T\min\{k, T\}}$ on the expected performance of $A$ on a distribution of networks with $m = nd$ edges.
\end{proof}

\begin{discussion}
The lower bound construction of Theorem \ref{thm.lowerbound} is tailored to the query model considered in this paper.  In particular, we do not assume that vertices are not sorted by degree, component size, etc.  However, the construction can be easily modified to be robust to various changes in the model, by (for example) adding edges with small weight so that the exhibited network $\calg$ becomes connected and/or regular.
\end{discussion}

\subsection*{Acknowledgments}
We thank Elchanan Mossel for helpful discussions.

\bibliographystyle{acmsmall}

\bibliography{sublinearInfluence}

\appendix
%

\section{Concentration Bounds}

For reference, we now provide the statement of the Chernoff bounds used throughout this paper.

\begin{lemma}
\label{lemma.multchernoff}
Let $X_i$ be $n$ i.i.d. Bernoulli random variables with expectation $\mu$ each. Define $X = \sum_{i=1}^{n}{X_i}$.
Then, 
\begin{itemize}
\item
For $0 <\lambda < 1$:$~\Pr[X < (1-\lambda)\mu n] < \exp(-\mu n\lambda^2/2)$.
\item
For $0 <\lambda < 1$:$~\Pr[X > (1+\lambda)\mu n] < \exp(-\mu n\lambda^2/4)$.
\item
For $\lambda \geq 1$:$~\Pr[X > (1+\lambda)\mu n] < \exp(- \mu n\lambda/3)$.      
\end{itemize}
\end{lemma}



\end{document}